%% file: main.tex
\documentclass[conference]{IEEEtran}
\IEEEoverridecommandlockouts
\usepackage{cite}
\usepackage{amsmath,amssymb,amsfonts}
\usepackage{algorithmic}
\usepackage{algorithm}
\usepackage{graphicx}
\usepackage[caption=false,font=footnotesize]{subfig}
\usepackage{textcomp}
\usepackage[dvipsnames,svgnames,table]{xcolor}
\usepackage{booktabs}
\usepackage{makecell}
\newcommand{\methodname}{SHADE}
\newcommand{\method}{\textbf{\texttt{\methodname}}}
\newcommand{\gray}[1]{\textcolor{gray}{#1}}
\newcommand{\ket}[1]{\lvert #1\rangle}
\newcommand{\bra}[1]{\langle #1\rvert}
\newcommand{\ketbra}[2]{\lvert #1\rangle\langle #2\rvert}
\newcommand{\venue}[1]{{\scriptsize\textcolor{gray!70!black}{[#1]}}}
\colorlet{bestcell}{red!13}
\colorlet{secondcell}{blue!10}
\newcommand{\best}[1]{\cellcolor{bestcell}\textbf{#1}}
\newcommand{\second}[1]{\cellcolor{secondcell}\textbf{#1}}
\DeclareRobustCommand{\swatch}[1]{{\setlength{\fboxsep}{0pt}\setlength{\fboxrule}{0.4pt}\raisebox{-0.2ex}{\fcolorbox{gray!80}{#1}{\rule{0pt}{1.8ex}\hspace{5ex}}}}}
\newcommand{\band}[1]{\gray{\textit{#1}}}
\definecolor{hdrA}{HTML}{8A1F3D}
\definecolor{hdrB}{HTML}{16479B}
\definecolor{hdrC}{HTML}{15683A}
\definecolor{hdrD}{HTML}{4D4D4D}
\newcommand{\hd}[2]{\cellcolor{#1}\textcolor{white}{\textbf{#2}}}
\definecolor{oursname}{HTML}{FBE5D0}

\newcommand{\std}[1]{$_{\pm\text{\scriptsize #1}}$}
\usepackage{tcolorbox}
\usepackage{ntheorem}
\theoremseparator{:}

\newtheorem{proposition}{Proposition}
\newtheorem{theorem}{Theorem}
\newtheorem{corollary}{Corollary}
\theoremstyle{nonumberplain}
\theorembodyfont{\normalfont}
\newtheorem{proof}{Proof}
\usepackage{hyperref}
\usepackage{balance}

\def\BibTeX{{\rm B\kern-.05em{\sc i\kern-.025em b}\kern-.08em
    T\kern-.1667em\lower.7ex\hbox{E}\kern-.125emX}}
\begin{document}

\title{Federated Unlearning of Quantum Data via Classical Shadows}

\author{
\IEEEauthorblockN{Zihao Ding\textsuperscript{1}, Jun Huang\textsuperscript{1}, and Cong-cong Xing\textsuperscript{2}}
\IEEEauthorblockA{\textsuperscript{1}Department of Electrical Engineering and Computer Science, South Dakota State University, Brookings, SD 57007, USA\\
\textsuperscript{2}Department of Mathematics and Computer Science, Nicholls State University, Thibodaux, LA 70310, USA\\
Email: zihao.ding@jacks.sdstate.edu, jun.huang@sdstate.edu, cong-cong.xing@nicholls.edu}
}

\maketitle

\begin{abstract}
Quantum federated learning trains a quantum neural network (QNN) across nodes that hold quantum data, such as the states of quantum sensors. To forget data, unlearning methods for quantum models replay them, and federated unlearning from stored updates removes whole nodes only. Quantum data cannot be replayed, since a state is consumed when measured and cannot be cloned. We propose \underline{Sha}dow-based \underline{D}ata \underline{E}rasure (\method{}), in which each node records a classical shadow of every training state, a few randomized single-qubit measurements taken at training time. The gradient and curvature of a QNN on a state are expectations of observables, so \method{} estimates them for the forget data from the shadows at any later parameters and removes their effect by a federated Newton correction. Only classical estimates cross the network. We prove that the shadow size needed for a target unlearning error grows exponentially in the size of the backward light cone of the measured observable and not in the number of qubits. In simulation with 50 nodes and on edge boards, \method{} matches the forget accuracy of retraining within 0.5 points, without replay and in under 0.2~s.
\end{abstract}

\begin{IEEEkeywords}
Quantum federated learning, machine unlearning, classical shadows, quantum neural networks, quantum sensor networks.
\end{IEEEkeywords}

\input{section/1_introduction}
\input{section/3_systemAndproblem}
\input{section/4_method}
\input{section/5_experiments}
\input{section/6_conclusion}

\balance
\input{main.bbl}

\end{document}

%% file: section/1_introduction.tex
\section{Introduction}
\label{sec:intro}

Quantum sensors and small quantum processors produce data that are quantum states, such as the probe states of a sensor network or the output states of a quantum simulator~\cite{Degen2017RMP}. In a quantum sensor network, for instance, each node holds the probe states of its own field and classifies the field from them~\cite{Zhang2021QST}. Sending these states to one site needs quantum links, which are short and lossy~\cite{Wehner2018Science}. Quantum federated learning (QFL) therefore trains a quantum neural network (QNN)~\cite{Cong2019NatPhys}, a parameterized circuit shared by all nodes, where the states are produced, and it exchanges only classical gradients~\cite{Chen2021Entropy,Chehimi2022ICASSP}. Networked systems~\cite{Wu2026COMST,Wu2026IoTJ,Wu2026TON,Wu2026Network,Ding2025IPCCC,Ding2026TAAS,Ding2026ICNC,Huang2020TWC,Huang2019WCM,Che2027PAKDD} keep learning after deployment~\cite{Wu2026TNSE,Wu2026ICDCSAlpha,Wu2026ARXIV1,Wu2026ARXIVPRISM,Wu2026arXivRouting,Wu2026arXivREBASE,Wu2026arXivCrystalMem}, and they must also be able to forget~\cite{Bourtoule2021SP,Ding2026MNET,Ding2026ICDCS,Pudasaini2026HPSRAgri,Ding2026ARXIVEASE,Ding2026arXivMUTE,Che2026arXiv}. A node may withdraw its data, or a sensor may turn out to have drifted for a period, and its states must then stop shaping the shared circuit.

Unlearning methods for quantum models evaluate the model on the forget data again. Gradient ascent, distillation, and pruning by the quantum Fisher information all measure the QNN on the forget states~\cite{Su2026PRApplied,Kurmanji2023NIPS}, and certified removal needs their gradient and curvature at the current parameters~\cite{Guo2020ICML}. Federated unlearning from stored updates avoids new measurements~\cite{Liu2021IWQoS,ZhangL2023TIFS}, but it removes whole nodes and cannot remove a window of time or a subset of samples. A drifted sensor, however, spoils only the states of one period, while the other states of its node remain sound. These methods share the assumption that the training data can be accessed again. Quantum data break this assumption. A state is consumed when it is measured, the no-cloning theorem forbids keeping a copy~\cite{Wootters1982Nature}, and a sensor cannot re-create a field that has since changed. What a node keeps from training are the outcomes of parameter-shift measurements~\cite{Schuld2019PRA}, and each outcome holds only for the parameters of its round, as Fig.~\ref{fig:record} shows.

The structure of a QNN, however, makes deletion without replay possible. The gradient and the curvature of a QNN on a state are expectation values of observables, linear in the state. A classical shadow, the record of random single-qubit Pauli measurements on copies of a state, predicts expectation values of observables chosen after the measurement~\cite{Huang2020NatPhys}. A node that records shadows of its states during training can thus estimate, at any later parameters, the gradient and curvature that the forget data contribute. The observable also sets the cost of this record. Under a local measurement, the observables of a shallow circuit act only on the backward light cone of the measured qubits~\cite{Cerezo2021NatCommun}, as Fig.~\ref{fig:cone} shows, so the record grows with the light cone and not with the number of qubits.

Motivated by these insights, we design Shadow-based Data Erasure (\method{}). Each node splits the copies of every state between training and shadow snapshots, and it keeps the snapshots on the light cone only. When a request arrives, the nodes estimate the gradient of the forget data and the curvature of the retained data from their snapshots at the current model. The server applies one Newton correction, and the snapshots of the forget data are erased. Only classical estimates cross the network, and no state is prepared again. The record spends a few training copies per state on the ability to forget it later.

Our contributions are summarized as follows:
\begin{itemize}
\item We formulate forgetting in QFL when forget states cannot be re-measured, as choosing the copies recorded per state so that the correction meets a given error target.
\item We propose \method{}, which records classical shadows on the light cone in training and undoes the forget data by a federated Newton step built from these shadows.
\item We prove that the shadow estimates are unbiased at any parameters and that the snapshots needed per state grow exponentially in the size of the light cone and not in the number of qubits, and we bound the parameter distance between the corrected and the retrained model.
\item In simulation with 50 nodes and on edge boards, \method{} matches the forget accuracy of retraining within 0.5 points, without replaying states and in under 0.2~s.
\end{itemize}

%% file: section/3_systemAndproblem.tex
\section{System Model and Problem Formulation}
\label{sec:system}

\begin{figure}[t]
  \centering
  \includegraphics[width=0.95\columnwidth]{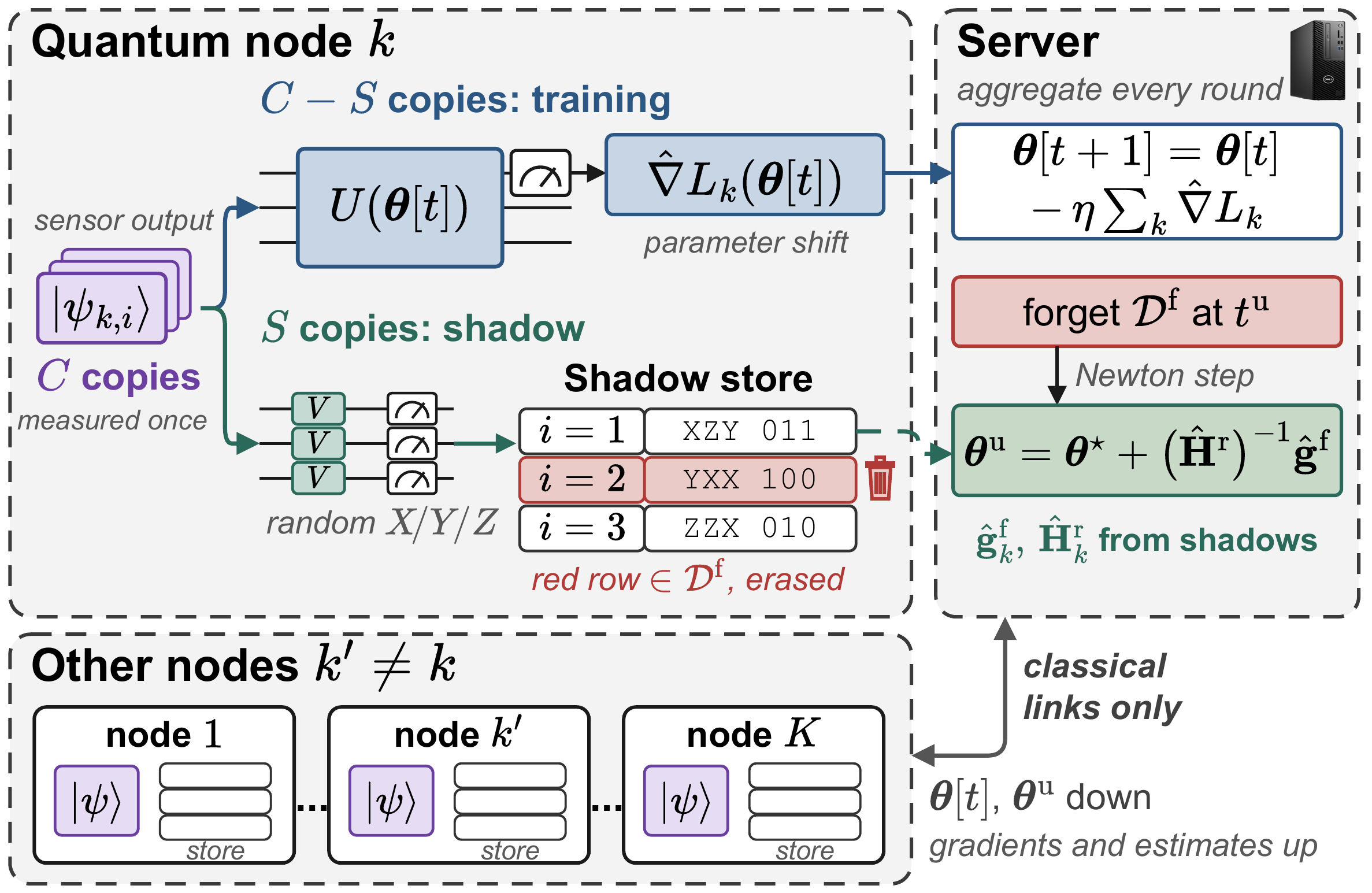}
  \caption{System model of federated unlearning on quantum data.}
  \label{fig:system}
\end{figure}

\begin{figure*}[t]
  \centering
  \begin{minipage}[t]{0.48\textwidth}
    \centering
    \includegraphics[width=\linewidth]{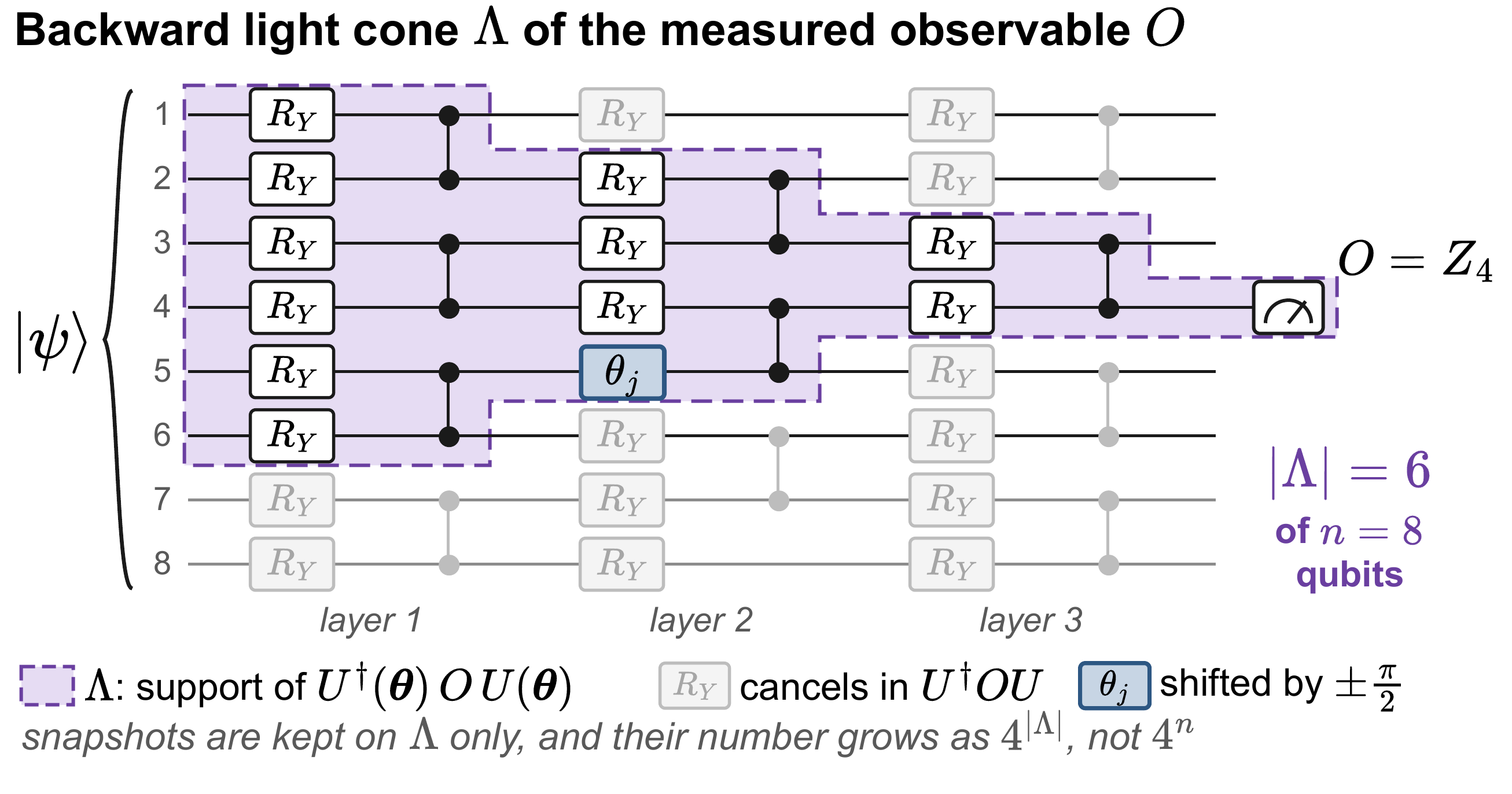}
    \caption{Backward light cone of the measured observable in a layered QNN.}
    \label{fig:cone}
  \end{minipage}
  \hfill
  \begin{minipage}[t]{0.48\textwidth}
    \centering
    \includegraphics[width=\linewidth]{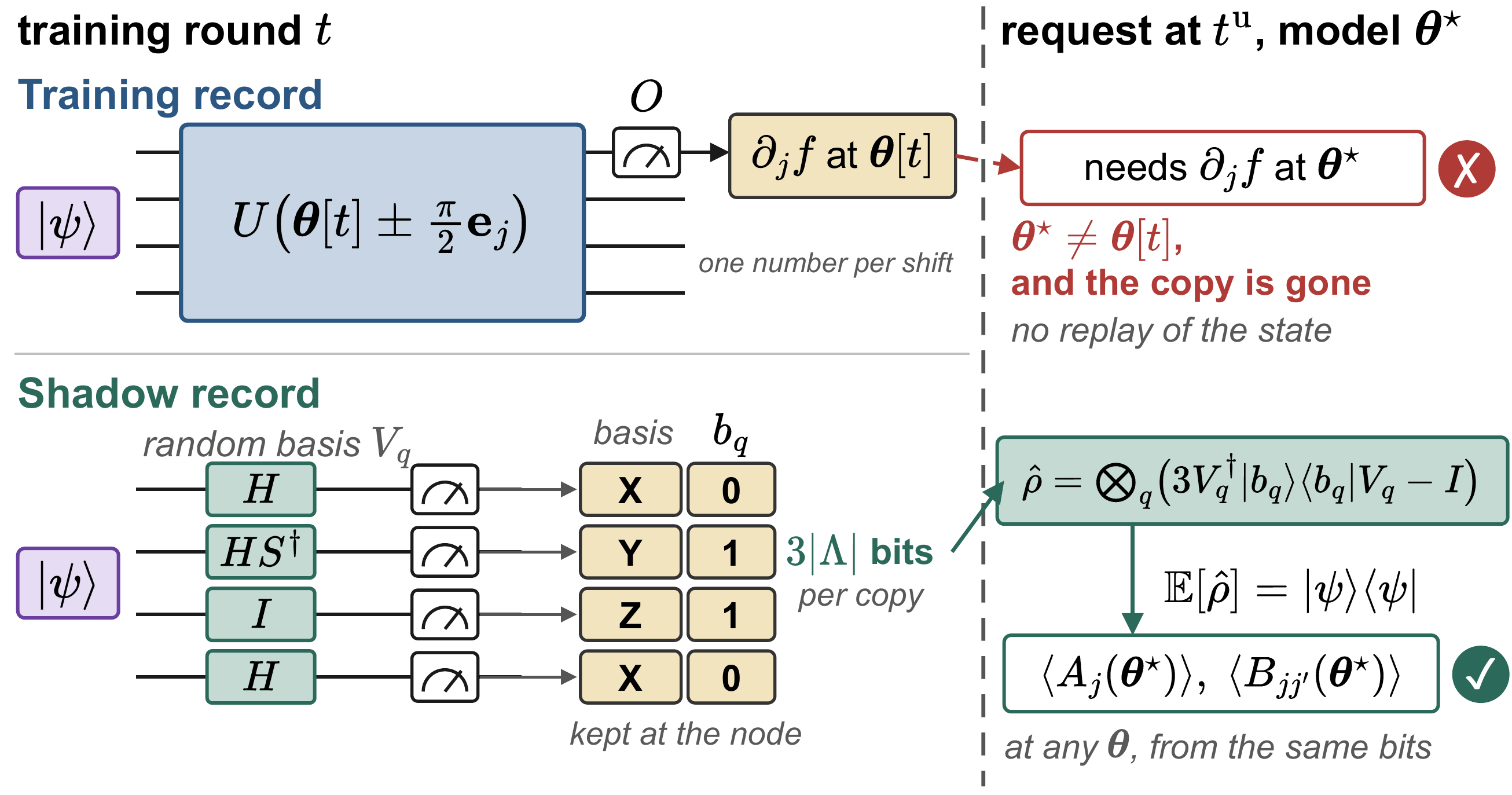}
    \caption{Parameter-shift records fix the parameters, shadows do not.}
    \label{fig:record}
  \end{minipage}
\end{figure*}

\subsection{System Model}
\label{subsec:system-model}

We consider a server and $K$ quantum nodes connected by classical links, as illustrated in Fig.~\ref{fig:system}. Node $k$ holds a dataset $\mathcal{D}_k$ of labeled states $(\ket{\psi_{k,i}},y_{k,i})$ with $y_{k,i}\in\{-1,+1\}$, produced by its sensor. We let $\mathcal{D}=\bigcup_k\mathcal{D}_k$ denote the union of all local datasets. State $i$ of node $k$ is produced in round $\tau_{k,i}$, and the sensor prepares $C$ copies of it. The copies are finite because a sensor prepares them only while the field it probes stays the same. A copy can be measured once and is then gone, and the node has no quantum memory to keep a copy for later. Fig.~\ref{fig:system} shows node $k$ splitting the copies of a state between training and a classical record. The node keeps classical storage and a classical link to the server, over which it uploads $P$ numbers per round during training.

The nodes train a QNN, an $n$-qubit circuit $U(\boldsymbol{\theta})$ with parameters $\boldsymbol{\theta}\in\mathbb{R}^{P}$, in which each parameter sets the angle of one rotation generated by a Pauli operator. Its output on a state is the expectation of an observable $O$ that acts on a few qubits and satisfies $\|O\|\le1$,
\begin{equation}
    f(\psi;\boldsymbol{\theta})\;=\;\bra{\psi}U^{\dagger}(\boldsymbol{\theta})\,O\,U(\boldsymbol{\theta})\ket{\psi}.
    \label{eq:output}
\end{equation}
For a loss $\ell(f,y)$ and a weight $\mu>0$, the objective on a set $\mathcal{D}'\subseteq\mathcal{D}$ is
\begin{equation}
    L(\boldsymbol{\theta};\mathcal{D}')\;=\;\sum_{(\psi,y)\in\mathcal{D}'}\ell\bigl(f(\psi;\boldsymbol{\theta}),y\bigr)+\frac{\mu}{2}\|\boldsymbol{\theta}\|^{2}.
    \label{eq:objective}
\end{equation}
Node $k$ holds the local objective $L_k$, the sum of $\ell$ over $\mathcal{D}_k$ plus $\frac{\mu}{2K}\|\boldsymbol{\theta}\|^{2}$, so that $\sum_k L_k(\boldsymbol{\theta})=L(\boldsymbol{\theta};\mathcal{D})$.

In round $t$, the server broadcasts $\boldsymbol{\theta}[t]$. Node $k$ estimates $\nabla L_k(\boldsymbol{\theta}[t])$ by the parameter-shift rule $\partial f/\partial\theta_j=\frac{1}{2}\bigl[f(\psi;\boldsymbol{\theta}+\frac{\pi}{2}\mathbf{e}_j)-f(\psi;\boldsymbol{\theta}-\frac{\pi}{2}\mathbf{e}_j)\bigr]$, where $\mathbf{e}_j$ is the $j$-th unit vector, and it measures fresh copies of its states on the shifted circuits. The server updates
\begin{equation}
    \boldsymbol{\theta}[t+1]\;=\;\boldsymbol{\theta}[t]-\eta\sum_{k=1}^{K}\hat{\nabla}L_k(\boldsymbol{\theta}[t])
    \label{eq:update}
\end{equation}
with the step size $\eta$ and the estimates $\hat{\nabla}L_k$ that the nodes upload. The server holds no state and sees only these classical uploads. In a round, a node uses a subset of its states and spends $2P$ shifted circuits on each, so training spends $C^{\textrm{tr}}$ copies of each state over all rounds, a number that grows with the rounds in which the state is used.

\subsection{Problem Formulation}
\label{subsec:problem-formulation}

At round $t^{\textrm{u}}$, a request names a forget set $\mathcal{D}^{\textrm{f}}\subseteq\mathcal{D}$. We consider three kinds. A node deletion names $\mathcal{D}_k$ when node $k$ withdraws. A window deletion names the states of node $k$ with $\tau_{k,i}$ in an interval of rounds, such as those produced while its sensor drifted. A sample deletion names any subset, such as the states whose owners withdraw consent. The part of $\mathcal{D}^{\textrm{f}}$ at node $k$ is $\mathcal{D}^{\textrm{f}}_k$, and the retain set is $\mathcal{D}^{\textrm{r}}=\mathcal{D}\setminus\mathcal{D}^{\textrm{f}}$ with parts $\mathcal{D}^{\textrm{r}}_k$. The server must serve a request before its next update, since every round on the old model spreads the forget data further. Let $\boldsymbol{\theta}^{\star}$ denote the model at $t^{\textrm{u}}$, which we take as a minimizer of $L(\boldsymbol{\theta};\mathcal{D})$. The target is the retrained model $\boldsymbol{\theta}^{\textrm{r}}=\arg\min_{\boldsymbol{\theta}}L(\boldsymbol{\theta};\mathcal{D}^{\textrm{r}})$, the model that the federation would hold had it never seen $\mathcal{D}^{\textrm{f}}$, against which membership inference finds nothing to detect. Retraining is not possible, since the copies of the retained states were consumed as well. The forget states cannot be measured again either, and the outcomes that training left behind hold only at the past parameters $\boldsymbol{\theta}[t]\neq\boldsymbol{\theta}^{\star}$, as Fig.~\ref{fig:record} shows. Each outcome is an eigenvalue of $O$ after a shifted circuit, and averaging them over the rounds that used a state estimates its gradient at a mixture of past models. Neither the gradient nor the curvature at $\boldsymbol{\theta}^{\star}$ follows from it.

A node could also keep its uploads $\hat{\nabla}L_k(\boldsymbol{\theta}[t])$ of every round, as federated unlearning from stored updates does. Each upload sums over all states that the node used in that round, so it cannot separate the forget states of a window or a sample deletion from the retained states of the same node. Subtracting the uploads removes the whole node, and the retained states of that node are lost with it.

What a node can do is record classical data from some copies of each state during training, before any request is known. Let $S$ copies of every state be set aside for this record. The server computes an unlearned model $\boldsymbol{\theta}^{\textrm{u}}$ from the records, and we seek the smallest record that serves the requests,
\begin{equation}
\begin{split}
    (\mathrm{P}):\;\;\min\;\;
    & S \\
    \mathrm{s.t.}\;\;
    & (\mathrm{C}_1):\;\Pr\bigl[\|\boldsymbol{\theta}^{\textrm{u}}-\boldsymbol{\theta}^{\textrm{r}}\|\le\varepsilon^{\textrm{u}}\bigr]\ge1-\delta, \\
    & (\mathrm{C}_2):\;S\le C-C^{\textrm{tr}}, \\
    & (\mathrm{C}_3):\;\boldsymbol{\theta}^{\textrm{u}}\text{ uses no copy of }\mathcal{D}^{\textrm{f}}\text{ after }t^{\textrm{u}},
\end{split}
\label{eq:problem}
\end{equation}
where $\varepsilon^{\textrm{u}}$ is the error target and $\delta$ the failure probability. The forget set is unknown during training, so every state is recorded, and $S$ sets both the copies taken from training and the storage per node. The distance in $(\mathrm{C}_1)$ also bounds every output. The parameter-shift rule gives $|\partial f/\partial\theta_j|\le\|O\|\le1$, so $f(\psi;\cdot)$ is $\sqrt{P}$-Lipschitz, and $|f(\psi;\boldsymbol{\theta}^{\textrm{u}})-f(\psi;\boldsymbol{\theta}^{\textrm{r}})|\le\sqrt{P}\,\varepsilon^{\textrm{u}}$ holds for every state, forgotten or retained.

%% file: section/4_method.tex
\section{\texorpdfstring{\method{}}{SHADE}: Shadow-based Data Erasure}
\label{sec:method}

Problem~$(\mathrm{P})$ in Eq.~\eqref{eq:problem} cannot be solved directly. A correction toward $\boldsymbol{\theta}^{\textrm{r}}$ needs the gradient of the forget data and the curvature of the retained data at $\boldsymbol{\theta}^{\star}$, which is known only at $t^{\textrm{u}}$. The records must be taken before any request, so they must serve every state and every future model, and they must stay small because every state needs one. To address these issues, we design \underline{\textbf{Sha}}dow-based \underline{\textbf{D}}ata \underline{\textbf{E}}rasure (\method{}), whose three modules are described in Sections~\ref{subsec:record}, \ref{subsec:estimate}, and~\ref{subsec:correct}. Fig.~\ref{fig:arch} shows where each module runs. The first runs at every node in every round, the second at every node once a request arrives, and the third at the server.

\subsection{Shadow Recording on the Light Cone}
\label{subsec:record}

This module records $S$ copies of every state during training. For copy $s$ of a state, the node draws for every qubit $q$ a single-qubit Clifford rotation $V_{q,s}$ uniformly from the three that map the $X$, $Y$, and $Z$ bases to the computational basis, and it measures the copy to obtain bits $b_{q,s}$. The snapshot
\begin{equation}
    \hat{\rho}_s\;=\;\bigotimes_{q=1}^{n}\bigl(3V_{q,s}^{\dagger}\ketbra{b_{q,s}}{b_{q,s}}V_{q,s}-I\bigr)
    \label{eq:snapshot}
\end{equation}
satisfies $\mathbb{E}[\hat{\rho}_s]=\ketbra{\psi}{\psi}$ over the random rotations and outcomes. Hence $\mathrm{Tr}[A\hat{\rho}_s]$ estimates $\bra{\psi}A\ket{\psi}$ for any observable $A$, including one chosen after the measurement.

The node does not keep all $n$ qubits. The backward light cone $\Lambda$ is the set of qubits that reach the support of $O$ through the gates of $U$, as drawn in Fig.~\ref{fig:cone}. Every gate outside $\Lambda$ cancels in $U^{\dagger}(\boldsymbol{\theta})OU(\boldsymbol{\theta})$ for every $\boldsymbol{\theta}$, and the factor of Eq.~\eqref{eq:snapshot} on a qubit outside the support of an observable has unit trace. The node therefore stores $(V_{q,s},b_{q,s})$ for $q\in\Lambda$ only, three bits per qubit and snapshot, or $3|\Lambda|S$ bits per state. With $|\Lambda|=6$ and $S=2000$, a state takes 4.5 kB. The snapshots are taken on copies that training does not use, so their randomness is independent of the trained model. They stay at the node and are never uploaded. The forget set is unknown during training, so every state gets the same $S$, set by Section~\ref{sec:analysis}.

On a real device, readout errors shrink the mean of $V_{q,s}^{\dagger}\ketbra{b_{q,s}}{b_{q,s}}V_{q,s}$ toward the maximally mixed state, and Eq.~\eqref{eq:snapshot} then no longer inverts it. The node calibrates this shrinkage per qubit by taking snapshots of the known state $\ket{0}^{\otimes n}$ in the same way. It replaces the factor 3 and the identity term of Eq.~\eqref{eq:snapshot} by the inverse of the calibrated map, which keeps the estimates unbiased as long as the readout errors do not depend on the measured state. The price is a larger variance, by at most the inverse square of the product of the shrinkage factors over $\Lambda$, and $S$ grows by the same factor.

\subsection{Estimation from Shadows}
\label{subsec:estimate}

This module runs at the nodes when the server broadcasts a request and $\boldsymbol{\theta}^{\star}$. Let $O(\boldsymbol{\theta})=U^{\dagger}(\boldsymbol{\theta})OU(\boldsymbol{\theta})$, so that $f(\psi;\boldsymbol{\theta})=\bra{\psi}O(\boldsymbol{\theta})\ket{\psi}$ by Eq.~\eqref{eq:output}. Applying the parameter-shift rule once and twice writes the derivatives of $f$ as expectations as well,
\begin{equation}
    \frac{\partial f}{\partial\theta_j}=\bra{\psi}A_j(\boldsymbol{\theta})\ket{\psi},\qquad
    \frac{\partial^{2}f}{\partial\theta_j\partial\theta_{j'}}=\bra{\psi}B_{jj'}(\boldsymbol{\theta})\ket{\psi},
    \label{eq:shift}
\end{equation}
where $A_j(\boldsymbol{\theta})=\frac{1}{2}\bigl[O(\boldsymbol{\theta}+\frac{\pi}{2}\mathbf{e}_j)-O(\boldsymbol{\theta}-\frac{\pi}{2}\mathbf{e}_j)\bigr]$ and $B_{jj'}(\boldsymbol{\theta})$ combines the four operators shifted by $\pm\frac{\pi}{2}$ in both $\theta_j$ and $\theta_{j'}$ with weights $\pm\frac{1}{4}$. A shift changes an angle and not the layout of the gates, so all of these operators act on $\Lambda$ only.

For every state it holds, a node estimates $f$, $\nabla f$, and $\nabla^{2}f$ at $\boldsymbol{\theta}^{\star}$ by the median of means of $\mathrm{Tr}[O(\boldsymbol{\theta}^{\star})\hat{\rho}_s]$, $\mathrm{Tr}[A_j(\boldsymbol{\theta}^{\star})\hat{\rho}_s]$, and $\mathrm{Tr}[B_{jj'}(\boldsymbol{\theta}^{\star})\hat{\rho}_s]$ over its snapshots. It then forms the gradient $\nabla\ell=\ell'\,\nabla f$ and the Hessian $\nabla^{2}\ell=\ell''\,\nabla f\,\nabla f^{\top}+\ell'\,\nabla^{2}f$ of the loss, where $\ell'$ and $\ell''$ are the derivatives of $\ell$ in $f$ at the estimated output. The two factors of $\nabla f\,\nabla f^{\top}$ come from disjoint halves of the snapshots, so their errors are independent. Node $k$ sums
\begin{equation}
    \hat{\mathbf{g}}^{\textrm{f}}_k=\sum_{(\psi,y)\in\mathcal{D}^{\textrm{f}}_k}\hat{\nabla}\ell,\qquad
    \hat{\mathbf{H}}^{\textrm{r}}_k=\sum_{(\psi,y)\in\mathcal{D}^{\textrm{r}}_k}\hat{\nabla}^{2}\ell
    \label{eq:local}
\end{equation}
and uploads $P$ and $P(P+1)/2$ numbers, where a node without forget data uploads only $\hat{\mathbf{H}}^{\textrm{r}}_k$. Each evaluation contracts a snapshot with an operator on $\Lambda$, which takes classical time exponential in $|\Lambda|$ and no quantum hardware. With $P=24$ and $S=2000$, a state needs $(P+1)^{2}S\approx1.25\times10^{6}$ contractions on six qubits, all independent and thus parallel.

\subsection{Federated Newton Correction and Erasure}
\label{subsec:correct}

This module runs at the server and closes the request. Let $\mathbf{g}^{\textrm{f}}$ denote the sum of $\nabla\ell$ over $\mathcal{D}^{\textrm{f}}$ at $\boldsymbol{\theta}^{\star}$. Since $\boldsymbol{\theta}^{\star}$ minimizes $L(\boldsymbol{\theta};\mathcal{D})$, the retain objective has the gradient $\nabla L(\boldsymbol{\theta}^{\star};\mathcal{D}^{\textrm{r}})=-\mathbf{g}^{\textrm{f}}$ there. One Newton step on $L(\boldsymbol{\theta};\mathcal{D}^{\textrm{r}})$ from $\boldsymbol{\theta}^{\star}$ thus undoes the pull of the forget data along the curvature of the retained data. The server sums the uploads into $\hat{\mathbf{g}}^{\textrm{f}}=\sum_k\hat{\mathbf{g}}^{\textrm{f}}_k$ and $\hat{\mathbf{H}}^{\textrm{r}}=\sum_k\hat{\mathbf{H}}^{\textrm{r}}_k+\mu\mathbf{I}$, symmetrizes $\hat{\mathbf{H}}^{\textrm{r}}$, raises its eigenvalues below $\mu$ to $\mu$, and sets
\begin{equation}
    \boldsymbol{\theta}^{\textrm{u}}\;=\;\boldsymbol{\theta}^{\star}+\bigl(\hat{\mathbf{H}}^{\textrm{r}}\bigr)^{-1}\hat{\mathbf{g}}^{\textrm{f}}.
    \label{eq:newton}
\end{equation}
It broadcasts $\boldsymbol{\theta}^{\textrm{u}}$, the nodes erase the snapshots of $\mathcal{D}^{\textrm{f}}$, and training resumes from $\boldsymbol{\theta}^{\textrm{u}}$ on new states. The erasure is part of the deletion, since a shadow describes the state. Snapshots of $\mathcal{D}^{\textrm{f}}$ left at a node could serve to learn the forget data again.

The uploads $\hat{\mathbf{g}}^{\textrm{f}}_k$ and $\hat{\mathbf{H}}^{\textrm{r}}_k$ are sums over states, of the same kind as the gradients that training uploads in every round, so the correction discloses no more about a single state to the server than training does. Requests that arrive in the same round are served together, with $\mathcal{D}^{\textrm{f}}$ the union of their forget sets. A later request evaluates the snapshots at a model that the earlier correction computed from them, which breaks the independence that Section~\ref{sec:analysis} requires. \method{} thus splits each state's snapshots into equal groups and uses a fresh group per request, so the record serves as many requests as it has groups.

The correction costs one round of uploads and one $P\times P$ solve, and it meets $(\mathrm{C}_3)$ since it uses no copy of any state.

\begin{figure*}[t]
  \centering
  \begin{minipage}[t]{0.48\textwidth}
    \centering
    \includegraphics[width=\linewidth]{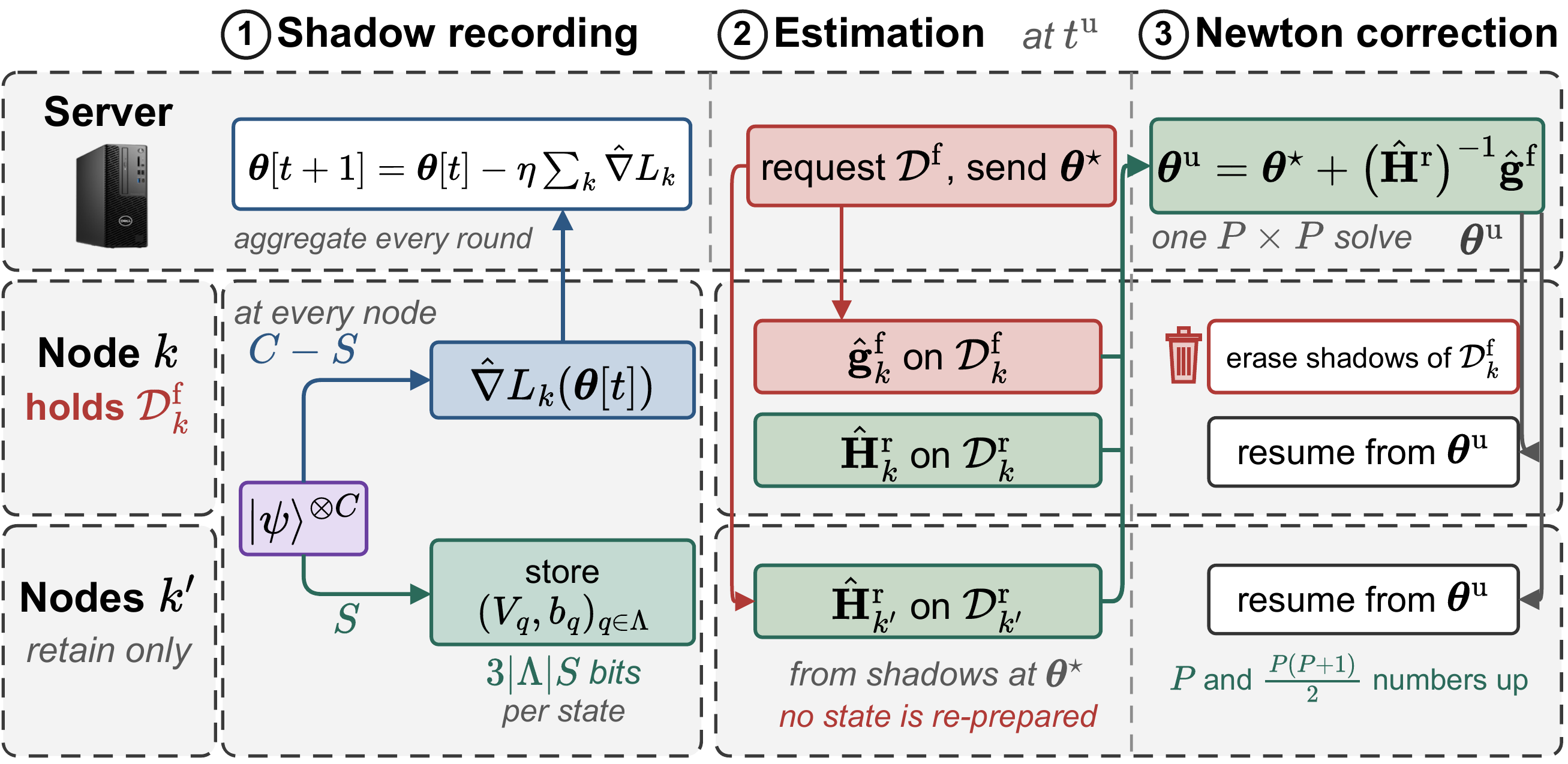}
    \caption{The architecture of \method{} at the nodes and the server.}
    \label{fig:arch}
  \end{minipage}
  \hfill
  \begin{minipage}[t]{0.48\textwidth}
    \centering
    \includegraphics[width=\linewidth]{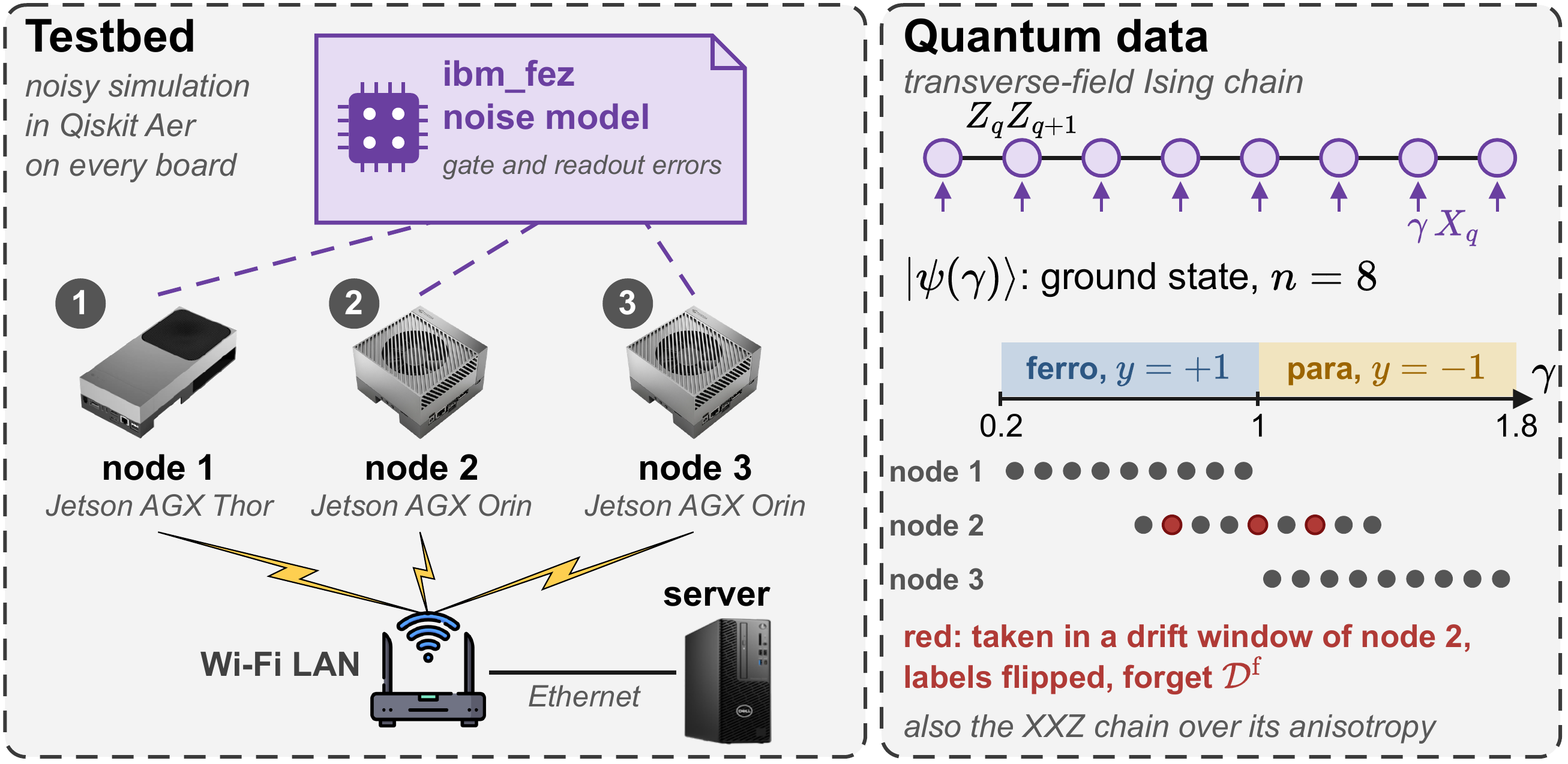}
    \caption{Testbed of edge boards with device noise, and the spin-chain data.}
    \label{fig:testbed}
  \end{minipage}
\end{figure*}

\begin{table*}[!t]
\centering
\setlength{\tabcolsep}{4pt}
\setlength{\aboverulesep}{0pt}
\setlength{\belowrulesep}{0pt}
\setlength{\extrarowheight}{0pt}
\renewcommand{\arraystretch}{0.88}
\caption{Forgetting, utility, and cost per method and request type. FA, RA, and TA are the accuracies in \% on the forget, retained, and held-out states, and FA should match Retrain. Dist and ODist are the parameter and output distances to Retrain, and MIA is the attack success rate in \%. Replay counts the copies of forget states used after the request, the storage is in kB per state, the uplink is in kB per node, and the time is in seconds. The testbed columns come from the node boards under the noise model of ibm\_fez. Methods marked \textup{\textdagger} are given fresh copies that a deployment does not have. A \swatch{bestcell} cell marks the best value and a \swatch{secondcell} cell the second best within each request type, excluding Retrain.}
\label{tab:main}
\resizebox{\linewidth}{!}{\begin{tabular}{l|ccc|ccc|cccc|cccc}
\Xhline{1.2pt}
 & \multicolumn{3}{c|}{\hd{hdrA}{Forgetting}} & \multicolumn{3}{c|}{\hd{hdrB}{Utility and privacy}} & \multicolumn{4}{c|}{\hd{hdrC}{Cost}} & \multicolumn{4}{c}{\hd{hdrD}{Testbed}}\\
\textbf{Method} & FA & Dist$\downarrow$ & ODist$\downarrow$ & RA$\uparrow$ & TA$\uparrow$ & MIA$\downarrow$ & Replay$\downarrow$ & Storage$\downarrow$ & Uplink$\downarrow$ & Time$\downarrow$ & FA & Dist$\downarrow$ & TA$\uparrow$ & Time$\downarrow$\\
\Xhline{1.2pt}
\multicolumn{15}{l}{\band{(a) Node deletion}}\\
\Xhline{0.6pt}
FedRecovery \venue{TIFS'23} & 31.4\std{21.9} & 0.288\std{0.044} & 0.0580\std{0.0344} & 99.8 & 93.6\std{2.2} & 50.2\std{1.0} & \best{0} & 4.90 & \best{$3.2{\times}10^{-4}$} & \best{0.01}\std{0.0006} & \second{22.9}\std{17.4} & 0.358\std{0.026} & 92.3\std{1.8} & \best{0.02}\std{0.004}\\
QMU$^{\dagger}$ \venue{PRApplied'26} & 52.4\std{1.9} & 0.217\std{0.040} & 0.0740\std{0.0011} & \best{100.0} & \best{96.1}\std{0.2} & \second{50.1}\std{0.7} & $3.9{\times}10^{5}$ & \best{$3.2{\times}10^{-3}$} & \second{0.02} & 99.1\std{0.002} & 51.4\std{1.2} & 0.268\std{0.040} & \second{95.3}\std{0.2} & 100\std{0.005}\\
SCRUB$^{\dagger}$ \venue{NeurIPS'23} & 49.1\std{1.0} & 0.203\std{0.042} & 0.0681\std{0.0041} & 99.6 & 95.4 & 50.2\std{0.2} & $2.7{\times}10^{6}$ & \best{$3.2{\times}10^{-3}$} & 2.84 & 693\std{0.008} & 47.8 & 0.259\std{0.043} & 95.1\std{0.2} & 703\std{0.06}\\
CR$^{\dagger}$ \venue{ICML'20} & \second{21.9}\std{2.5} & \best{0.192}\std{0.058} & \best{0.0151}\std{0.0063} & 99.2 & 92.2\std{0.4} & \best{49.7}\std{0.5} & $1.3{\times}10^{7}$ & \best{$3.2{\times}10^{-3}$} & 1.18 & 3226\std{0.009} & 20.1\std{3.1} & \best{0.192}\std{0.048} & 92.2\std{0.2} & 3227\std{0.007}\\
\Xhline{0.8pt}
\cellcolor{oursname}\textit{\textbf{\methodname{}}} & \best{26.4}\std{1.9} & \second{0.197}\std{0.063} & \second{0.0171}\std{0.0062} & \best{100.0} & \second{96.0}\std{0.4} & \second{50.1}\std{0.8} & \best{0} & 72.00 & 1.18 & \second{0.17}\std{0.002} & \best{24.6}\std{3.5} & \second{0.205}\std{0.077} & \best{95.8}\std{0.5} & \second{0.18}\std{0.007}\\
\Xhline{0.8pt}
\textit{Retrain}$^{\dagger}$ & 26.6\std{1.0} & 0.161\std{0.029} & 0.0012\std{0.0006} & 99.9 & 96.2\std{0.2} & 50.0\std{0.2} & $7.7{\times}10^{8}$ & $3.0{\times}10^{-3}$ & 37.63 & $1.9{\times}10^{5}$\std{1.0} & 24.8\std{2.0} & 0.170\std{0.014} & 96.0\std{0.2} & $1.9{\times}10^{5}$\std{0.9}\\
\Xhline{0.8pt}
\multicolumn{15}{l}{\band{(b) Window deletion}}\\
\Xhline{0.6pt}
FedRecovery \venue{TIFS'23} & 25.7\std{31.6} & 0.207\std{0.097} & 0.0987\std{0.0903} & 93.7\std{1.6} & 93.7\std{1.5} & 50.1 & \best{0} & 4.90 & \best{$8.0{\times}10^{-5}$} & \best{0.02}\std{0.007} & 25.2\std{35.4} & 0.288\std{0.142} & 91.8\std{2.5} & \best{0.02}\std{0.003}\\
QMU$^{\dagger}$ \venue{PRApplied'26} & 17.8\std{21.8} & 0.033\std{0.015} & 0.0055\std{0.0015} & 96.2\std{0.3} & \second{96.1}\std{0.2} & \best{49.9} & $9.8{\times}10^{4}$ & \best{$3.2{\times}10^{-3}$} & \second{0.02} & 25.0\std{0.002} & 12.3\std{17.7} & 0.034\std{0.012} & 95.9\std{0.2} & 25.5\std{0.01}\\
SCRUB$^{\dagger}$ \venue{NeurIPS'23} & \second{21.4}\std{24.5} & 0.037\std{0.012} & 0.0047\std{0.0035} & 96.1\std{0.5} & 96.0\std{0.5} & 50.1 & $2.5{\times}10^{6}$ & \best{$3.2{\times}10^{-3}$} & 2.90 & 619\std{0.008} & 20.7\std{29.5} & 0.041\std{0.012} & 96.0\std{0.8} & 629\std{0.07}\\
CR$^{\dagger}$ \venue{ICML'20} & 17.2\std{21.8} & \best{0.022}\std{0.012} & \best{0.0007}\std{0.0004} & \best{96.6}\std{0.4} & 95.9\std{0.4} & \best{49.9} & $1.3{\times}10^{7}$ & \best{$3.2{\times}10^{-3}$} & 1.20 & 3226\std{0.02} & \second{15.2}\std{23.6} & \best{0.031}\std{0.015} & \second{96.2}\std{0.5} & 3227\std{0.007}\\
\Xhline{0.8pt}
\cellcolor{oursname}\textit{\textbf{\methodname{}}} & \best{19.7}\std{21.8} & \second{0.026}\std{0.018} & \second{0.0009}\std{0.0005} & \second{96.5}\std{0.4} & \best{96.4}\std{0.4} & 50.1 & \best{0} & 72.00 & 1.20 & \second{0.18}\std{0.003} & \best{16.3}\std{23.6} & \second{0.033}\std{0.018} & \best{96.6}\std{0.6} & \second{0.18}\std{0.006}\\
\Xhline{0.8pt}
\textit{Retrain}$^{\dagger}$ & 20.0\std{24.5} & 0.152\std{0.022} & 0.0018\std{0.0006} & 96.7\std{0.4} & 96.5\std{0.4} & 50.0 & $7.8{\times}10^{8}$ & $3.0{\times}10^{-3}$ & 38.40 & $2.0{\times}10^{5}$\std{0.8} & 16.7\std{23.6} & 0.169\std{0.037} & 96.7\std{0.2} & $2.0{\times}10^{5}$\std{2.0}\\
\Xhline{0.8pt}
\multicolumn{15}{l}{\band{(c) Sample deletion}}\\
\Xhline{0.6pt}
FedRecovery \venue{TIFS'23} & 52.3\std{13.0} & 0.678\std{0.573} & 0.3305\std{0.3144} & 81.7\std{16.4} & 82.0\std{16.1} & 50.2\std{0.4} & \best{0} & 4.90 & \best{$8.0{\times}10^{-4}$} & \best{0.02}\std{0.006} & 53.9\std{15.5} & 1.037\std{0.611} & 74.8\std{17.8} & \best{0.02}\std{0.002}\\
QMU$^{\dagger}$ \venue{PRApplied'26} & 52.2\std{2.5} & 0.035\std{0.022} & 0.0043\std{0.0032} & 95.6 & 96.0\std{0.2} & 50.2\std{0.7} & $7.8{\times}10^{5}$ & \best{$3.2{\times}10^{-3}$} & \second{0.34} & 198\std{0.002} & 58.3\std{2.4} & 0.039\std{0.027} & 96.0\std{0.2} & 201\std{0.007}\\
SCRUB$^{\dagger}$ \venue{NeurIPS'23} & \second{51.6}\std{2.5} & 0.046\std{0.020} & 0.0149\std{0.0071} & 95.9\std{0.3} & \second{96.5}\std{0.2} & \best{49.9}\std{0.2} & $3.1{\times}10^{6}$ & \best{$3.2{\times}10^{-3}$} & 3.22 & 792\std{0.009} & 58.4\std{2.4} & 0.052\std{0.026} & \second{96.2}\std{0.2} & 803\std{0.09}\\
CR$^{\dagger}$ \venue{ICML'20} & 52.5\std{4.0} & \best{0.025}\std{0.015} & \best{0.0011}\std{0.0002} & \best{96.7}\std{0.2} & 96.1\std{0.2} & \second{50.0}\std{0.2} & $1.3{\times}10^{7}$ & \best{$3.2{\times}10^{-3}$} & 1.23 & 3226\std{0.009} & \second{51.1}\std{2.4} & \best{0.037}\std{0.025} & 96.1\std{0.2} & 3227\std{0.009}\\
\Xhline{0.8pt}
\cellcolor{oursname}\textit{\textbf{\methodname{}}} & \best{29.3}\std{4.0} & \second{0.029}\std{0.013} & \second{0.0015}\std{0.0003} & \second{96.6}\std{0.2} & \best{96.9}\std{0.2} & \second{50.0}\std{0.4} & \best{0} & 72.00 & 1.23 & \second{0.18}\std{0.007} & \best{30.1}\std{2.4} & \second{0.038}\std{0.013} & \best{96.5}\std{0.2} & \second{0.19}\std{0.003}\\
\Xhline{0.8pt}
\textit{Retrain}$^{\dagger}$ & 29.6\std{4.0} & 0.194\std{0.014} & 0.0012\std{0.0006} & 96.8\std{0.1} & 97.1\std{0.2} & 49.9\std{0.2} & $7.4{\times}10^{8}$ & $3.0{\times}10^{-3}$ & 38.40 & $1.9{\times}10^{5}$\std{1.1} & 30.2\std{2.4} & 0.173\std{0.017} & 96.7\std{0.2} & $1.9{\times}10^{5}$\std{1.7}\\
\Xhline{1.2pt}
\end{tabular}}
\end{table*}

\section{Theoretical Analysis}
\label{sec:analysis}

\begin{proposition}\label{prop:unbiased}
For $\boldsymbol{\theta}$ chosen independently of the snapshots of a state $\ket{\psi}$, the averages of $\mathrm{Tr}[O(\boldsymbol{\theta})\hat{\rho}_s]$, $\mathrm{Tr}[A_j(\boldsymbol{\theta})\hat{\rho}_s]$, and $\mathrm{Tr}[B_{jj'}(\boldsymbol{\theta})\hat{\rho}_s]$ over its snapshots are unbiased estimates of $f(\psi;\boldsymbol{\theta})$ and of its first and second derivatives.
\end{proposition}

\begin{proof}
On one qubit, measuring a state $\rho$ after a uniformly random $V_{q,s}$ yields $V_{q,s}^{\dagger}\ketbra{b_{q,s}}{b_{q,s}}V_{q,s}$ with mean $\frac{1}{3}(\rho+\mathrm{Tr}[\rho]\,I)$, and the map $X\mapsto3X-\mathrm{Tr}[X]\,I$ inverts this mean. The rotations and outcomes are drawn qubit by qubit, so the product in Eq.~\eqref{eq:snapshot} inverts the $n$-qubit mean, and $\mathbb{E}[\hat{\rho}_s]=\ketbra{\psi}{\psi}$ holds for any state, entangled or not. The trace is linear, and Eqs.~\eqref{eq:output} and~\eqref{eq:shift} write $f$ and its derivatives as expectations of $O(\boldsymbol{\theta})$, $A_j(\boldsymbol{\theta})$, and $B_{jj'}(\boldsymbol{\theta})$. Since $\boldsymbol{\theta}$ does not depend on the snapshots, the expectation over the snapshots may be taken at fixed $\boldsymbol{\theta}$.\end{proof}

\begin{theorem}\label{thm:size}
Let the snapshots of every state be combined by the median of means. There is an absolute constant $\alpha$ such that with
\begin{equation}
    S\;=\;\Bigl\lceil\alpha\,\frac{4^{|\Lambda|}}{\varepsilon^{2}}\,\log\frac{|\mathcal{D}|(P+1)^{2}}{\delta}\Bigr\rceil
    \label{eq:size}
\end{equation}
snapshots per state, the estimates of $f$, $\partial f/\partial\theta_j$, and $\partial^{2}f/\partial\theta_j\partial\theta_{j'}$ for all states in $\mathcal{D}$ and all $j,j'$ are within $\varepsilon$ of their values at a fixed $\boldsymbol{\theta}$ with probability at least $1-\delta$. $S$ does not depend on $n$.
\end{theorem}

\begin{proof}
By Section~\ref{subsec:estimate}, $O(\boldsymbol{\theta})$, $A_j(\boldsymbol{\theta})$, and $B_{jj'}(\boldsymbol{\theta})$ act on $\Lambda$ only. Each is a sum of unitary conjugates of $O$ whose weights have absolute values summing to at most one, so its operator norm is at most $\|O\|\le1$. For random Pauli measurements, the variance of $\mathrm{Tr}[A\hat{\rho}_s]$ for an operator $A$ on $m$ qubits is at most $4^{m}\|A\|^{2}$, and the median of means of $\alpha\,4^{m}\|A\|^{2}\log(1/\delta')/\varepsilon^{2}$ snapshots is within $\varepsilon$ of the mean with probability $1-\delta'$. Each state needs $1+P+P(P+1)/2\le(P+1)^{2}$ estimates, and a union bound over the $|\mathcal{D}|(P+1)^{2}$ estimates with $\delta'=\delta/(|\mathcal{D}|(P+1)^{2})$ gives Eq.~\eqref{eq:size}.\end{proof}

According to \textbf{Theorem}~\ref{thm:size}, the depth of the circuit rather than its width sets the record. Each brickwork layer adds at most two qubits to $\Lambda$ and thus multiplies $S$ by at most 16. The same locality lets a classical processor compute a shallow QNN from measured data of its inputs. It does not make the device dispensable, since the states exist only at the node and every snapshot and training outcome must be taken there. A deeper circuit is trained on the device as before, and only its record grows, by the factor above per layer.

\begin{proposition}\label{prop:error}
Suppose that $L(\cdot;\mathcal{D}^{\textrm{r}})$ is $\kappa$-strongly convex with an $M$-Lipschitz Hessian on a ball that contains $\boldsymbol{\theta}^{\star}$ and $\boldsymbol{\theta}^{\textrm{r}}$, and let $\mathbf{H}^{\textrm{r}}=\nabla^{2}L(\boldsymbol{\theta}^{\star};\mathcal{D}^{\textrm{r}})$. If $\|\hat{\mathbf{g}}^{\textrm{f}}-\mathbf{g}^{\textrm{f}}\|\le\varepsilon^{\mathbf{g}}$ and the matrix in Eq.~\eqref{eq:newton} satisfies $\|\hat{\mathbf{H}}^{\textrm{r}}-\mathbf{H}^{\textrm{r}}\|\le\varepsilon^{\mathbf{H}}<\kappa$, then
\begin{equation}
    \|\boldsymbol{\theta}^{\textrm{u}}-\boldsymbol{\theta}^{\textrm{r}}\|\;\le\;\frac{M\|\mathbf{g}^{\textrm{f}}\|^{2}}{2\kappa^{3}}+\frac{\varepsilon^{\mathbf{g}}}{\kappa-\varepsilon^{\mathbf{H}}}+\frac{\varepsilon^{\mathbf{H}}\|\mathbf{g}^{\textrm{f}}\|}{\kappa(\kappa-\varepsilon^{\mathbf{H}})}.
    \label{eq:error}
\end{equation}
\end{proposition}

\begin{proof}
Let $\bar{\boldsymbol{\theta}}=\boldsymbol{\theta}^{\star}+(\mathbf{H}^{\textrm{r}})^{-1}\mathbf{g}^{\textrm{f}}$. Since $\nabla L(\boldsymbol{\theta}^{\star};\mathcal{D}^{\textrm{r}})=-\mathbf{g}^{\textrm{f}}$, $\bar{\boldsymbol{\theta}}$ is one exact Newton step on $L(\cdot;\mathcal{D}^{\textrm{r}})$, so $\|\bar{\boldsymbol{\theta}}-\boldsymbol{\theta}^{\textrm{r}}\|\le\frac{M}{2\kappa}\|\boldsymbol{\theta}^{\star}-\boldsymbol{\theta}^{\textrm{r}}\|^{2}$. Strong convexity gives $\kappa\|\boldsymbol{\theta}^{\star}-\boldsymbol{\theta}^{\textrm{r}}\|\le\|\nabla L(\boldsymbol{\theta}^{\star};\mathcal{D}^{\textrm{r}})\|=\|\mathbf{g}^{\textrm{f}}\|$, which yields the first term. Next, $\boldsymbol{\theta}^{\textrm{u}}-\bar{\boldsymbol{\theta}}=(\hat{\mathbf{H}}^{\textrm{r}})^{-1}(\hat{\mathbf{g}}^{\textrm{f}}-\mathbf{g}^{\textrm{f}})+(\hat{\mathbf{H}}^{\textrm{r}})^{-1}(\mathbf{H}^{\textrm{r}}-\hat{\mathbf{H}}^{\textrm{r}})(\mathbf{H}^{\textrm{r}})^{-1}\mathbf{g}^{\textrm{f}}$. Weyl's inequality gives $\|(\hat{\mathbf{H}}^{\textrm{r}})^{-1}\|\le1/(\kappa-\varepsilon^{\mathbf{H}})$ and $\|(\mathbf{H}^{\textrm{r}})^{-1}\|\le1/\kappa$, and the triangle inequality completes the bound.\end{proof}

\begin{corollary}\label{cor:square}
For $\ell=(f-y)^{2}/2$ and $\varepsilon\le1$, the event of \textbf{Theorem}~\ref{thm:size} at $\boldsymbol{\theta}^{\star}$ implies $\varepsilon^{\mathbf{g}}\le4\sqrt{P}\,|\mathcal{D}^{\textrm{f}}|\,\varepsilon$ and $\varepsilon^{\mathbf{H}}\le7P\,|\mathcal{D}^{\textrm{r}}|\,\varepsilon$.
\end{corollary}

\begin{proof}
Here $\nabla\ell=(f-y)\nabla f$ and $\nabla^{2}\ell=\nabla f\,\nabla f^{\top}+(f-y)\nabla^{2}f$, with $|f-y|\le2$ and every derivative of $f$ at most one in magnitude. Each entry of the estimated gradient is thus within $\varepsilon(1+\varepsilon)+2\varepsilon\le4\varepsilon$, and each entry of the estimated Hessian within $3\varepsilon+4\varepsilon$. Norms add a factor $\sqrt{P}$ for a gradient and $P$ for the Frobenius norm of a Hessian, which bounds its spectral norm, and Eq.~\eqref{eq:local} adds the errors of its states. For $\mu\le\kappa$, symmetrizing and raising the eigenvalues to $\mu$ is the Frobenius projection onto a convex set that contains $\mathbf{H}^{\textrm{r}}$, which does not increase the error.\end{proof}

According to \textbf{Proposition}~\ref{prop:error}, the first term is the price of one Newton step and does not shrink with $S$, while by \textbf{Corollary}~\ref{cor:square} the other two are linear in $\varepsilon$. Hence $(\mathrm{C}_1)$ holds once $\varepsilon$ is a fixed fraction of the gap $\varepsilon^{\textrm{u}}-M\|\mathbf{g}^{\textrm{f}}\|^{2}/(2\kappa^{3})$, with $S$ from Eq.~\eqref{eq:size}. When the gap is negative, the corrected model needs further training on new states.

%% file: section/5_experiments.tex
\section{Performance Evaluation}
\label{sec:exp}

\subsection{Experimental Setup}
\label{subsec:setup}

\begin{figure*}[!t]
\centering
\subfloat[Snapshots per state.]{\includegraphics[width=0.24\textwidth]{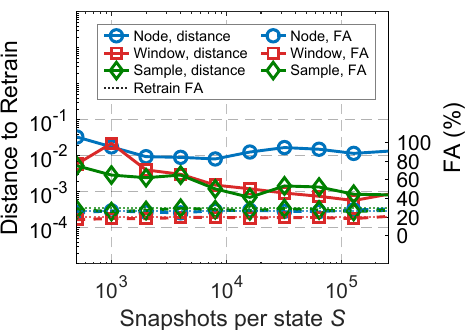}\label{fig:snap}}
\hfill
\subfloat[Circuit layers.]{\includegraphics[width=0.24\textwidth]{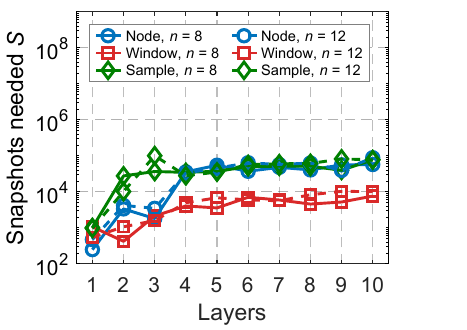}\label{fig:depth}}
\hfill
\subfloat[Copy split.]{\includegraphics[width=0.24\textwidth]{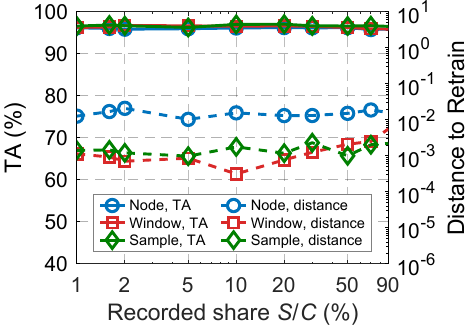}\label{fig:split}}
\hfill
\subfloat[Forget-set size.]{\includegraphics[width=0.24\textwidth]{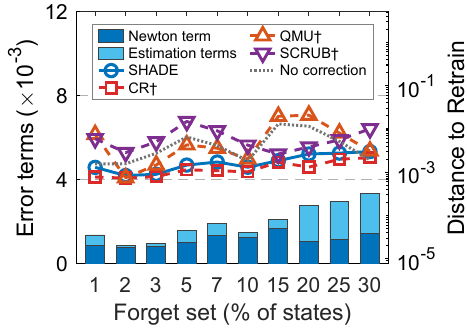}\label{fig:size}}
\caption{Results of \method{}. (a) Distance to Retrain and forget accuracy as the snapshots per state grow. (b) Snapshots needed for a fixed distance as the number of layers grows. (c) Test accuracy and distance to Retrain as the share of recorded copies grows. (d) Distance to Retrain as the forget set grows, with the Newton and estimation terms of \method{}.}
\label{fig:results}
\end{figure*}

We simulate a server and $K=50$ nodes in PennyLane with shot sampling at $n=8$ qubits. The states are ground states of the transverse-field Ising chain $-\sum_{q}(Z_qZ_{q+1}+\gamma X_q)$ at 400 values of $\gamma\in[0.2,1.8]$, labeled by the ferromagnetic phase $\gamma<1$ or the paramagnetic phase $\gamma>1$, and of the XXZ chain across its anisotropy, following the TensorFlow Quantum datasets on a finer grid. Each node holds the 8 states of one range of the coupling. The QNN is the circuit of Fig.~\ref{fig:cone} with $P=24$, $O=Z_4$, and $|\Lambda|=6$. It is trained on $\ell=(f-y)^{2}/2$ with $\mu=0.01$ and $\eta=0.05$ for 60 rounds, on 8 states per node and round with 20 shots per shifted circuit. Each state has $C=2\times10^{6}$ copies, and $S=3.2\times10^{4}$ are recorded unless stated otherwise. The request arrives at $t^{\textrm{u}}=60$. Node deletion removes one node, and sample deletion removes 5\% of all states. Window deletion removes the fifth of a node's states produced while its sensor drifted, whose labels are flipped. On the testbed in Fig.~\ref{fig:testbed}, three nodes run on a Jetson AGX Thor and two Jetson AGX Orin boards, and the server on a desktop, over a Wi-Fi LAN. Each board simulates the states, the training shots, and the snapshots of its node in Qiskit Aer under the noise model of the 156-qubit IBM processor ibm\_fez, which carries its gate and readout errors.

We measure forgetting by the accuracy $\mathrm{FA}$ on the forget states, which should match that of Retrain, and by the distance $\|\boldsymbol{\theta}^{\textrm{u}}-\boldsymbol{\theta}^{\textrm{r}}\|$. The output distance $\mathrm{ODist}$ is the root mean square of $f(\psi;\boldsymbol{\theta}^{\textrm{u}})-f(\psi;\boldsymbol{\theta}^{\textrm{r}})$ over the training and held-out states. Utility is the accuracy $\mathrm{RA}$ on the retained states and $\mathrm{TA}$ on held-out states, and privacy is the success rate $\mathrm{MIA}$ of a membership inference attack for QNNs on the forget states. The cost counts the copies of forget states used after the request, the snapshot storage per state, the uplink per node, and the time of the correction. We compare \method{} with four baselines. FedRecovery removes the contribution of a node from its stored updates without new measurements, so it removes every node that holds a forget state. QMU runs gradient ascent on the forget states, SCRUB distills the trained model while pushing its outputs on the forget states away, and CR applies an exact Newton correction. These three need fresh copies of the forget states, which a deployment does not have, and we grant them to show what replay achieves. Retrain trains from scratch on fresh copies of the retained states.

\subsection{Experimental Results}
\label{subsec:results}

\begin{figure}[t]
\centering
\subfloat[Number of qubits.]{\includegraphics[width=0.485\linewidth]{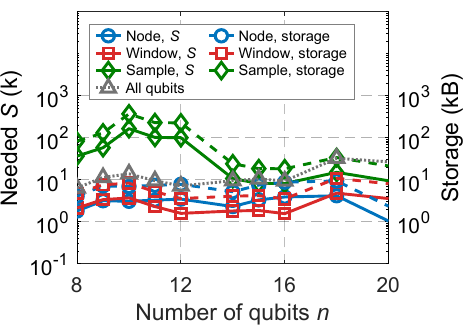}\label{fig:qubits}}
\hfill
\subfloat[Number of nodes.]{\includegraphics[width=0.485\linewidth]{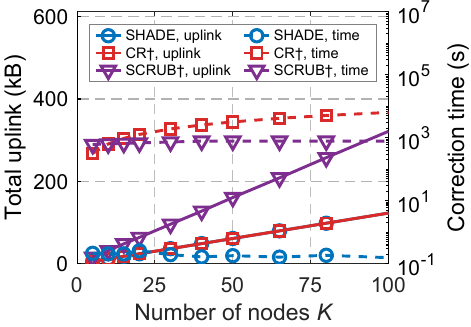}\label{fig:nodes}}
\caption{Scalability of \method{}. (a) Snapshots needed and storage per state as the number of qubits grows at three layers. (b) Uplink and correction time as the number of nodes grows.}
\label{fig:scalability}
\end{figure}

Table~\ref{tab:main} compares the methods. \method{} keeps the forget accuracy within 0.5 points of Retrain for all request types, in simulation and on the testbed. Its distance to Retrain is second only to CR, which uses exact derivatives from fresh copies. FedRecovery removes whole nodes, so under window and sample deletion its distance is at least seven times that of \method{}. Under sample deletion, QMU, SCRUB, and CR leave the forget accuracy near 52\%, against 29.6\% for Retrain. The attack success rate stays near 50\% for every method. \method{} replays no copy and corrects the model in under 0.2~s, against 3226~s for CR. Its uplink equals that of CR, and its cost is a classical record of 72~kB per state at each node.

In Fig.~\ref{fig:results}, the distance under window and sample deletion falls six- to sevenfold as $S$ grows. Under node deletion it stays near $10^{-2}$, as expected when the Newton term of Eq.~\eqref{eq:error} dominates. The snapshots needed are of the same order at $n=8$ and $n=12$, as \textbf{Theorem}~\ref{thm:size} predicts. The test accuracy stays above 95.7\% for recorded shares up to 90\%. The distance of \method{} stays within about twice that of CR as the forget set grows, and from 20\% on the estimation terms exceed the Newton term. In Fig.~\ref{fig:scalability}, node and window deletion need at most $4.9\times10^{3}$ snapshots up to 20 qubits. The uplink grows by 1.23~kB per node, and \method{} corrects in about 0.2~s at any number of nodes, against up to 6452~s for CR.

%% file: section/6_conclusion.tex
\section{Conclusion}
\label{sec:conclusion}

We studied how a QFL deployment forgets quantum data that can no longer be measured. \method{} records classical shadows on the light cone of the measured observable during training and removes the forget data by a federated Newton correction estimated from them, so that no state has to be prepared again. With 50 nodes, in simulation and on edge boards, it matches retraining in forget accuracy within 0.5 points and corrects the model in under 0.2~s.